\documentclass[12pt,letterpaper]{article}
\usepackage{amsmath,amsfonts,amsthm,amssymb,stmaryrd,bm,cite} 
\usepackage{relsize,tabls}
\allowdisplaybreaks  

\theoremstyle{plain}

\numberwithin{equation}{section}
\newtheorem{thm}{Theorem}[section]
\newtheorem{cor}[thm]{Corollary}
\newtheorem{lem}[thm]{Lemma}

\newenvironment{exam}[1]
{\begin{flushleft}\textbf{Example #1}.\enspace}%
{\end{flushleft}}

\newcommand{\positive}{{\mathbb N}}
\newcommand{\complex}{{\mathbb C}}
\newcommand{\real}{{\mathbb R}}
\newcommand{\cscript}{{\mathcal C}}
\newcommand{\dscript}{{\mathcal D}}
\newcommand{\kscript}{{\mathcal K}}
\newcommand{\lscript}{{\mathcal L}}
\newcommand{\vhat}{\widehat{v}}
\newcommand{\fbar}{\overline{f}}
\newcommand{\xunderbar}{\underline{x}}
\newcommand{\ctimes}{\mathrel{\mathlarger\cdot}}

\newcommand{\ab}[1]{\left|#1\right|}
\newcommand{\brac}[1]{\left\{#1\right\}}
\newcommand{\paren}[1]{\left(#1\right)}
\newcommand{\sqbrac}[1]{\left[#1\right]}
\newcommand{\elbows}[1]{{\left\langle#1\right\rangle}}

\begin{document}

\title{CURVATURE AND\\QUANTUM MECHANICS\\ON COVARIANT CAUSAL SETS
}
\author{S. Gudder\\ Department of Mathematics\\
University of Denver\\ Denver, Colorado 80208, U.S.A.\\
sgudder@du.edu
}
\date{}
\maketitle

\begin{abstract}
This article begins by reviewing the causal set approach in discrete quantum gravity. In our version of this approach a special role is played by covariant causal sets which we call $c$-causets. The importance of $c$-causets is that they support the concepts of a natural distance function, geodesics and curvature in a discrete setting. We then discuss curvature in more detail. By considering $c$-causets with a maximum and minimum number of paths, we are able to find $c$-causets with large and small average curvature. We then briefly discuss our previous work on the inflationary period when the curvature was essentially zero. Quantum mechanics on $c$-causets is considered next. We first introduce a free wave equation for $c$-causets. We then show how the state of a particle with a specified mass (or energy) can be derived from the wave equation. It is demonstrated for small examples that quantum mechanics predicts that particles tend to move toward vertices with larger curvature.
\end{abstract}

\section{Covariant Causal Sets}  
This article is based upon the following four guidelines. The universe is:
(1)\enspace Discrete,
(2) Structured,
(3)\enspace Expanding,
(4)\enspace Quantum Mechanical.
Guidelines 2, 3 and 4 are well-established and do not need discussion. However, Guideline~1
 is fairly unconventional so we shall briefly consider it. There is already evidence that the universe is discrete. We know that energy comes in discrete packets or quanta that are integer multiples of Planck's constant $h$. Also, electric charge only exists in integer multiples of the electron charge $e$ (or $\pm e/3$, $\pm 2e/3$ if you include quarks). What has not been experimentally observed is a discreteness of space and time. This may be due to the extreme smallness of candidates for elementary lengths and times such as the Planck length of about $10^{-33}$cm. and Planck time of about $10^{-43}$sec. It should also be mentioned that postulating discreteness avoids infinities and singularities that have plagued quantum field theory and general relativity theory.
 
 We shall describe the structure of the universe by a causal set or causet \cite{bdghs03,hen09,sor03,sur11,wc15}. Mathematically, a \textit{causet} is a finite partially ordered set $(x,<)$. For $a,b\in x$, we interpret $a<b$ as meaning that $b$ is in the causal future of $a$. If $a<b$ and there is no $c\in x$ with $a<c<b$ we say that $a$ \textit{is a parent of} $b$ and write $a\prec b$. Denoting the cardinality of $x$ by $\ab{x}$, a \textit{labeling} of $x$ is a map $\ell\colon x\to\brac{1,2,\ldots ,\ab{x}}$ such that $a<b$ implies $\ell (a)<\ell (b)$. A labeling of $x$ may be considered a ``birth order'' of the vertices of $x$. A \textit{covariant causet} ($c$-\textit{causet}) is a causet that has a unique labeling \cite{gud13,gud14,gudJ14,gud15}. In this article we shall only model possible universes by $c$-causets.
 
A \textit{path} in a $c$-causet $x$ is a finite sequence $a_1a_2\cdots a_n$ with $a_1\prec a_2\prec\cdots\prec a_n$. The \textit{height}
$h(a)$ of $a\in x$ is the cardinality, minus one, of a longest path in $x$ that ends with $a$. Two vertices $a,b$ are \textit{comparable} if
$a<b$ or $b<a$. It is shown in \cite{gud14} that a causet $x$ is a $c$-causet if and only if $a,b\in x$ are comparable whenever $h(a)\ne h(b)$. We call the set
\begin{equation*}
S_j(x)=\brac{a\in x\colon h(a)=j}
\end{equation*}
the $j$th \textit{shell} of $x$, $j=0,1,2,\ldots\,$. Letting $s_j(x)=\ab{S_j(x)}$, $j=0,1,2,\ldots ,k$, we call $(s_0(x),s_1(x),\ldots ,s_k(x))$ the
\textit{shell sequence} of $x$. A $c$-causet is uniquely determined by its shell sequence \cite{gud14}. Conversely, any finite sequence of positive integers is the shell sequence of a unique $x$-causet. We usually assume that $s_0(x)=1$ and the vertex labeled $1$ represents the big bang.

Let $\omega =a_1a_2\cdots a_n$ be a path in $x$ where $a_j\in\positive$ are the labels of the vertices. The \textit{length} of $\omega$ is
\begin{equation*}
\lscript (\omega )=\sqbrac{\sum _{j=1}^{n-1}(a_{j+1}-a_j)^2}^{1/2}
\end{equation*}
A \textit{geodesic} from $a$ to $b$ where $a<b$, is a path from $a$ to $b$ of smallest length. For $a<b$, we define the \textit{distance from}
$a$ \textit{to} $b$ to be $d(a,b)=\lscript (\omega )$ where $\omega$ is a geodesic from $a$ to $b$. It is shown in \cite{gud14} that if
$a<b<c$, then $d(a,c)\le d(a,b)+d(b,c)$. This shows that the triangle inequality holds for $d(a,b)$ when applicable. In this way, $d(a,b)$ is a weak type of metric. It is also shown in \cite{gud14} that a subpath of a geodesic is a geodesic. Of course, if $a<b$ then there is at least one geodesic from $a$ to $b$. The \textit{curvature} $K(a)$ of $a\in x$ is the number of geodesics, minus one, from the vertex labeled $1$ to
$a$ \cite{gud14,gudJ14,gud15}. The \textit{average curvature} is
\begin{equation*}
\cscript (x)=\sum\brac{K(a)\colon a\in x}/\ab{x}
\end{equation*}
We call a path from $1$ to a vertex in the top shell a \textit{maximal path} because such a path cannot be extended.

\begin{exam}{1}  
Let $x$ be the $c$-causet with shell sequence $(1,2,1)$. Using semicolons to separate shells, we can label the vertices by $(1;2,3;4)$. There are two maximal paths $1-2-4$ and $1-3-4$. Both of these paths have length $\sqrt{5}$ so they are both geodesics. We conclude that $d(1,4)=\sqrt{5}$ and $K(1)=-1$, $K(2)=K(3)=0$, $K(4)=1$. The average curvature becomes $\cscript (x)=0$.\qquad\qedsymbol
\end{exam}  

The shell sequence determines the ``shape'' or geometry of $x$. We view a $c$-causet $x$ as a framework or scaffolding of a possible universe. The vertices represent tiny cells that may or may not be occupied by a particle. Presumably these cells have four-dimensional Planckian volume, but we need not commit ourselves here. This geometry gives the kinematics of the system. The dynamics is described in terms of paths in $x$. It is natural to assume that particles tend to move along geodesics. Thus, they tend to collect around vertices with large curvature. (We shall have more to say about this in Section~4.) In this way, the curvature distribution of a universe determines the mass (energy) distribution. Notice that this is exactly opposite to the traditional approach in general relativity.

\section{Curvature} 
This section discusses the curvature distribution for $c$-causets. First of all, we are interested in finding $c$-causets with large curvatures. After all, our own universe has stars, black holes, galaxies and galaxy clusters. These certainly suggest the presence of huge, albeit local, curvatures even though our universe in the large is essentially flat. Although it is difficult to determine curvatures for arbitrary $c$-causets, it is easy to find the number of maximal paths. There are exceptions, but roughly speaking the more maximal paths we have, the more geodesics we have and hence, the more curvature. We then ask the following question. Given $n\in\positive$, which $c$-causets $x$ with
$\ab{x}=n$ have the largest number of maximal paths?

If $x$ has shell sequence $(1,s_1(x),s_2(x),\ldots ,s_k(x))$, the number of maximal paths is $s_1(x)s_2(x)\cdots s_k(x)$. Now
\begin{equation*}
n=\ab{x}-1=\sum _{i=1}^ks_i(x)
\end{equation*}
gives a partition of $n$ and we want to maximize $\Pi _{i=1}^ks_i(x)$. A partition of $n$, $n=\alpha _1+\alpha _2+\cdots +\alpha _k$,
$\alpha _i\in\positive$, is \textit{large} if $\alpha _1\alpha _2\cdots\alpha _k$ is a maximum among all partitions of $n$. The question posed at the end of the previous paragraph reduces to finding large partitions of $n$ because such partitions would determine shell sequence of
$c$-causets with the largest number of maximal paths.

To gain some intuition about this purely combinatorial problem, consider partitions of $10$. (Recall that the order of the numbers in a partition is immaterial.) A few partitions of $10$ are:
\begin{equation*}
10=9+1=8+2=7+3=6+4=5+5
\end{equation*}
The products in these partitions are: $9,16,21,24,25$. However, we get larger products if we consider partitions with three terms:
\begin{equation*}
10=5+3+2=4+3+3
\end{equation*}
The products become $30$ and $36$. In fact, the partitions with largest products are $3+3+4$ and $3+3+2+2$.

\begin{exam}{2}  
The large partitions for the first few positive integers are:
$2=2$, $3=3$, $4=4=2+2$, $5=3+2$, $6=3+3$, $7=3+4=3+2+2$, $8=3+3+2$, $9=3+3+3$, $10=3+3+4=3+3+2+2$, $11=3+3+3+2$, $12=3+3+3+3$, $13=3+3+3+4=3+3+3+2+2$, $14=3+3+3+3+2$, $15=3+3+3+3+3$, $16=3+3+3+3+4=3+3+3+3+2+2$.\quad\qedsymbol
\end{exam}  

After examining Example~2, the reader can easily make a conjecture given by the statement of Theorem~2.2. Our proof relies on the following simple lemma.

\begin{lem}       
\label{lem21}
If $\alpha,\beta\in\real$ satisfy $\alpha\ge 2$, $\beta >2$, then $\alpha +\beta<\alpha\beta$.
\end{lem}
\begin{proof}
We have that
\begin{align*}
0&<(\alpha -1)(\beta -2)=\alpha\beta -\beta -2\alpha +2
\intertext{Hence,}
\alpha +\beta&\le\alpha +\beta +(\alpha -2)<\alpha\beta\qedhere
\end{align*}
\end{proof}

\begin{thm}       
\label{thm22}
Let $n\in\positive$ with $n\ge 2$.\newline
If $n\equiv 0\pmod{3}$, then $n=3m$ and the large partition of $n$ is:
\begin{equation*}
n=3+3+\cdots +3\ (m\,3s)
\end{equation*}
If $n\equiv 2\pmod{3}$, then $n=3m+2$ and the large partition of $n$ is:
\begin{equation*}
n=3+3+\cdots +3+2\ (m\,3s)
\end{equation*}
If $n\equiv 1\pmod{3}$, then $n=3m+1=3(m-1)+4$ and the large partitions of $n$ are:
\begin{equation*}
n=3+3+\cdots +3+4=3+3+\cdots +3+2+2\ (m-1\,3s)
\end{equation*}
\end{thm}
\begin{proof}
Take a partition of $n$, $n=\alpha _1+\alpha _2+\cdots +\alpha _k$. If $\alpha _1\ge 5$, then $\alpha _1=\beta +\gamma$, $\beta\ge 2$,
$\gamma >2$ and $\alpha _1<\beta\gamma$ by Lemma~\ref{lem21}. Hence, we obtain another partition of $n$
\begin{equation*}
n=\beta +\gamma +\alpha _2+\cdots +\alpha _k
\end{equation*}
where $\alpha _1\alpha _2\cdots\alpha _k<\beta\gamma\alpha _2\cdots\alpha _k$. Continue this process until we arrive at a partition
$n=\beta _1+\beta _2+\cdots +\beta _s$ where $\beta _i$ is 2 or 3. If there are more than two 2s, we have, say
\begin{equation*}
n=\beta _1+\beta _2+\cdots +\beta _t+2+2+2
\end{equation*}
But then $n=\beta _1+\beta _2+\cdots +\beta _t+3+3$ with
\begin{equation*}
\beta _1\beta _2\cdots\beta _t\ctimes 8<\beta _1\beta _2\cdots\beta _t\ctimes 9
\end{equation*}
which gives a larger product. Again, continue this process until we arrive at a partition of $n$ with fewer than three 2s and the rest are 3s. This gives a large partition. We have then obtained the three cases in the theorem.
\end{proof}

\begin{cor}       
\label{cor23}
The number of maximal paths in a $c$-causet with the largest number of such paths follow the sequence:
$2,3,2^2,2\ctimes 3,3^2,2^2\ctimes 3,2\ctimes 3^2,3^3,2^2\ctimes 3^2,2\ctimes 3^3,\ldots$
\end{cor}

It follows from Theorem~\ref{thm22} that $c$-causets whose shell sequences consist of all 3s or all 3s and a 2 or all 3s and two 2s (or a 4) have the largest number of maximal paths. (Unlike partitions, the positions of the 2s are important.) Although we have not proved this, we conjecture that besides a few exceptions, such $c$-causets have the largest average curvature. The many examples we have examined all substantiate this conjecture. We now present three of them.

\begin{exam}{3}  
Suppose $x$ has shell sequence $(1,3,3,3,3,3,3,3)$. This $c$-causet has 22 vertices and $3^7=2187$ maximal paths. The following table summarizes the distances and curvatures for $x$.
\bigskip

\begin{tabular}{c|c|c|c|c|c|c|c|c|c|c|c|c}
$i$&1&2&3&4&5&6&7&8&9&10&11&12\\
\hline
$d(1,i)$&0&1&2&3&$\sqrt{8}$&$\sqrt{13}$&$\sqrt{18}$&$\sqrt{17}$&$\sqrt{22}$&$\sqrt{27}$&$\sqrt{26}$&$\sqrt{31}$\\
\hline
$K(i)$&$-1$&0&0&0&0&1&0&2&2&0&5&3\\
\end{tabular}
\vskip 2pc
\begin{tabular}{c|c|c|c|c|c|c|c|c|c|c}
$i$&13&14&15&16&17&18&19&20&21&22\\
\hline
$d(1,i)$&$\sqrt{36}$&$\sqrt{35}$&$\sqrt{40}$&$\sqrt{45}$&$\sqrt{44}$&$\sqrt{49}$&$\sqrt{54}$&$\sqrt{53}$&$\sqrt{58}$&$\sqrt{63}$\\
\hline
$K(i)$&0&9&4&0&14&5&0&20&6&0\\
\noalign{\bigskip}
\multicolumn{11}{c}%
{\textbf{Table 1 (Distances and Curvatures)}}\\
\end{tabular}
\bigskip
\noindent We conclude that this $c$-causet has average curvature
\begin{equation*}
\cscript (x)=\frac{70}{22}\approx 3.18\qquad\qedsymbol
\end{equation*}
\end{exam} 
\medskip

\begin{exam}{4}  
Suppose $x$ has shell sequence $(1,2,2,2,2,2,2,2,2,2,2)$. This $c$-causet has 21 vertices and $2^{10}=1024$ maximal paths. The next table summarizes the distances and curvatures for $x$.
\bigskip

\begin{tabular}{c|c|c|c|c|c|c|c|c|c|c|c}
$i$&1&2&3&4&5&6&7&8&9&10&11\\
\hline
$d(1,i)$&0&1&2&$\sqrt{5}$&$\sqrt{8}$&$\sqrt{9}$&$\sqrt{12}$&$\sqrt{13}$&$\sqrt{16}$&$\sqrt{17}$&$\sqrt{20}$\\
\hline
$K(i)$&$-1$&0&0&1&0&2&0&3&0&4&0\\
\end{tabular}
\vskip 2pc
\begin{tabular}{c|c|c|c|c|c|c|c|c|c|c}
$i$&12&13&14&15&16&17&18&19&20&21\\
\hline
$d(1,i)$&$\sqrt{21}$&$\sqrt{24}$&$\sqrt{25}$&$\sqrt{28}$&$\sqrt{29}$&$\sqrt{32}$&$\sqrt{33}$&$\sqrt{36}$&$\sqrt{37}$&$\sqrt{40}$\\
\hline
$K(i)$&5&0&6&0&7&0&8&0&9&0\\
\noalign{\bigskip}
\multicolumn{11}{c}%
{\textbf{Table 2 (Distances and Curvatures)}}\\
\end{tabular}
\bigskip
\noindent The average curvature of this $c$-causet is
\begin{equation*}
\cscript (x)=\frac{44}{21}\approx 2.1\qquad\qedsymbol
\end{equation*}
\end{exam} 

\begin{exam}{5}  
Suppose $x$ has shell sequence $(1,4,4,4,4,4)$. This $c$-causet has 21 vertices and $4^5=1024$ maximal paths. The next table now summarizes the distances and curvatures for $x$.
\bigskip

\begin{tabular}{c|c|c|c|c|c|c|c|c|c|c|c}
$i$&1&2&3&4&5&6&7&8&9&10&11\\
\hline
$d(1,i)$&0&1&2&3&4&$\sqrt{13}$&$\sqrt{18}$&$\sqrt{25}$&$\sqrt{32}$&$\sqrt{27}$&$\sqrt{34}$\\
\hline
$K(i)$&$-1$&0&0&0&0&1&0&1&0&0&2\\
\end{tabular}
\vskip 2pc
\begin{tabular}{c|c|c|c|c|c|c|c|c|c|c}
$i$&12&13&14&15&16&17&18&19&20&21\\
\hline
$d(1,i)$&$\sqrt{41}$&$\sqrt{48}$&$\sqrt{43}$&$\sqrt{50}$&$\sqrt{57}$&$\sqrt{64}$&$\sqrt{59}$&$\sqrt{66}$&$\sqrt{73}$&$\sqrt{80}$\\
\hline
$K(i)$&2&0&2&5&3&0&8&9&4&0\\
\noalign{\bigskip}
\multicolumn{11}{c}%
{\textbf{Table 3 (Distances and Curvatures)}}\\
\end{tabular}
\bigskip
\noindent The average curvature of this $c$-causet is
\begin{equation*}
\cscript (x)=\frac{30}{21}\approx 1.71\qquad\qedsymbol
\end{equation*}
\end{exam} 

We have previously considered $c$-causets with high curvature. Let us now examine the other extreme which is low curvature. As before, these should be obtained from $c$-causets with the least number of maximal paths. We first eliminate uninteresting $c$-causets that can have a single vertex in a shell above the first shell. We say that a partition of $n$ is \textit{trivial} if it has the form $n=n$ or
$n=m+1+1+\cdots +1$. A partition of $n$ is \textit{small} if it is nontrivial and has the form $n=\alpha _1+\alpha _2+\cdots +\alpha _k$ where $\alpha _1\alpha _2\cdots\alpha _k$ is minimal among all nontrivial partitions of $n$. Again, to get some intuition about small partitions, we consider some examples.

\begin{exam}{6}  
Some nontrivial partitions of 10 are
\begin{equation*}
10=2+3+5=5+5=6+4=7+3=8+2
\end{equation*}
The products of the terms are $30,25,24,21$ and 16. The smallest is $10=8+2$. The small partitions of 6 and 7 are:
$6=4+2$, $7=5+2$. These motivate the following theorem.\qquad\qedsymbol
\end{exam} 

\begin{thm}       
\label{thm24}
If $n\ge 4$, then its unique small partition is $n=(n-2)+2$.
\end{thm}
\begin{proof}
Suppose $m<n$ and $n=(n-m)+m$ is a nontrivial partition of $n$. Then $m\ge 2$ and $n-m\ge 2$. Suppose $m>2$ and $n-m>2$. Since $n>m+2$ we have
\begin{align*}
nm-2n&=n(m-2)>(m+2)(m-2)=m^2-4
\intertext{Hence,}
(n-m)m&=nm-m^2>2n-4=2(n-2)
\end{align*}
We conclude that the partition $n=(n-m)+m$ has larger product then the partition $n=(n-2)+2$ if $m>2$. Hence $n=(n-2)+2$ is the smallest partition for two term partitions. Now suppose
\begin{equation*}
n=\alpha +\beta +\sqbrac{n-(\alpha +\beta )}
\end{equation*}
is a nontrivial three term partition. Then from Lemma~\ref{lem21} we have that $\alpha\beta >\alpha +\beta$. It follows that
\begin{equation*}
(n-\alpha -\beta )\alpha\beta\ge (n-\alpha -\beta )(\alpha +\beta )
\end{equation*}
But $n=(n-\alpha -\beta )+(\alpha +\beta )$ is a two term nontrivial partition of $n$, so by our previous work we have that
\begin{equation*}
(n-\alpha -\beta )(\alpha +\beta )>(n-2)2
\end{equation*}
We now continue this process. For example, let
\begin{equation*}
n=\alpha +\beta +\gamma +\sqbrac{n-(\alpha +\beta +\gamma )}
\end{equation*}
be a nontrivial four term partition of $n$. We can form the three term partition
\begin{equation*}
n=(\alpha +\beta )+\gamma +\sqbrac{n-(\alpha +\beta +\gamma )}
\end{equation*}
and as before
\begin{equation*}
(n-\alpha -\beta -\gamma )\alpha\beta\gamma\ge (n-\alpha -\beta -\gamma )(\alpha +\beta )\gamma
\end{equation*}
which reduces to our previous case.
\end{proof}

The next two examples show that $c$-causets with shell sequences corresponding to small partitions have average curvature zero.

\begin{exam}{7}   
Consider a $c$-causet with shell sequence $(1,2,n)$. The maximal paths $1-2-4$ and $1-3-4$ are geodesics so $K(4)=1$. For the vertices $j=5,6,\ldots ,n+3$ we have the two paths $1-2-j$, $1-3-j$. If these two paths have the same length, then
\begin{equation*}
1+(j-2)^2=4+(j-3)^2
\end{equation*}
But this implies that $j=4$ which is a contradiction. Hence, only one of these two paths is a geodesic (it happens to be the second). Thus $K(j)=0$ for $j=5,6,\ldots ,n+3$. It follows that the average curvature is zero.\qquad\qedsymbol
\end{exam} 

\begin{exam}{8}   
Consider a $c$-causet with shell sequence $(1,6,2)$. The vertices are labeled $1,2,\ldots ,9$ where 8 and 9 are in the top shell. The paths to 8 are $1-2-8,\ldots ,1-7-8$ and the paths to 9 are $1-2-9,\ldots ,1-7-9$. It is easy to check that $1-4-8$ and $1-5-8$ are the geodesics to 8 and $1-5-9$ is the only geodesic to 9. Hence, $K(8)=1$, $K(1)=-1$ and $K(j)=0$, $j\ne -1,8$. It follows that the average curvature is zero. A similar analysis holds for any $c$-causet with shell sequence $(1,n,2)$ $n\ge 2$ in which case the average curvature is again
zero.\qquad\qedsymbol
\end{exam} 

Although they are useful for illustrating extreme cases of curvature, the $c$-causets considered in the last five examples are not suitable for describing our universe. According to Guideline~3 of the Introduction, our universe is expanding. From the $c$-causet viewpoint this means expanding both ``vertically'' and ``horizontally.'' The $c$-causets in Examples 3, 4 and 5 are not expanding ``horizontally'' while those in Examples 7 and 8 are not expanding ``vertically.''

Another property of our universe is that it is flat in the large which means that its average curvature is essentially zero. We have shown in previous work \cite{gud15} that this occurs for an expanding universe if the expansion is exponential. We have also shown that this zero average curvature does not hold for slower than exponential expansion \cite{gud15}. The simplest $c$-causet with an exponential expansion has shell sequence $(1,2,2^2,\ldots ,2^n)$. This is called the \textit{inflationary period} and we believe that this period ends at about $n=308$ \cite{gud15}. After this period, the system enters a multiverse period of parallel universes. All of the universes share a common inflationary period but then they separate and evolve on different paths. Which particular history path a universe takes depends on probabilities determined by quantum mechanics \cite{gudJ14,gud15}. (This quantum mechanics of the macroscopic picture is different, but possibly related to the quantum mechanics of the microscopic picture considered in the next section.) We have reason to believe that our particular universe has a pulsating growth \cite{gud15}.

We now briefly discuss the concept of time in this model. There are actually two types of time, geometrical time and chronological time. The geometrical time is built into the causal relation $a<b$ where $b$ is in the causal future of $a$. The chronological time is universe dependent and is the natural time that a universe ``ticks off'' after it completes a shell. After a universe completes a particular shell it goes into a new cycle until it completes a new shell by adding vertices one at a time. Presumably each cycle takes a Planck instant of about $10^{-43}$ sec.\ to complete. We estimate that our universe has about $10^{60}$ shells and that each shell on average has about $10^{25}$ vertices. In summary, one can say that the universe does not evolve \textit{in} time, the universe \textit{makes} the time.

\section{Quantum Mechanics} 
This section discusses quantum mechanics on a $c$-causet. In order to develop a quantum theory we shall need an analogue of
Schr\"odinger's or Dirac's equations which describe the evolution of quantum states. Recall that Schr\"odinger's equation has the form
\begin{equation}         
\label{eq31}
i\,\frac{\partial}{\partial t}\,\psi (t,\xunderbar )=H\psi (t,\xunderbar )=\sqbrac{\nabla ^2+V(\xunderbar )}\psi (t,\xunderbar )
\end{equation}
where $\psi$ is a wave function, $H$ the Hamiltonian, $\xunderbar$ is the 3-dimensional position, $\nabla ^2$ the Laplacian and $V$ is the potential energy of the system. To find a discrete analogue of \eqref{eq31} we have two obstacles to overcome. First, we do not have the concept of a derivative. Second, a partial derivative is actually a directional derivative and we do not have directions. We circumvent these problems by replacing a derivative by a difference and a direction by a path.

Let $x$ be an arbitrary $c$-causet. For $a\in x$, we say that a path $\omega$ in $x$ \textit{contains} $a$ and write $a\in\omega$ if $\omega$ has the form $a_1a_2\cdots a_n$ where $a_i=a$ for some $i$. For a path
\begin{equation*}
\omega =\cdots aa_1a_2\cdots a_{n-1}b\cdots
\end{equation*}
we define the \textit{covariant difference operator} $\nabla _\omega$ as follows. The \textit{domain} of $\nabla _\omega$ is
\begin{equation*}
\dscript (\nabla _\omega )=\brac{v\colon x\times x\to\complex\colon v(a,b)=0\hbox{ if }a\nless b\hbox{ and }a,b\notin\omega}
\end{equation*}
and for $v\in\dscript (\nabla _\omega )$ we have
\begin{equation*}
\nabla _\omega v(a,b)=d(a,a_{n-1})v(a,b)-d(a,b)v(a,a_{n-1})
\end{equation*}
Notice that if we suitably redefine the distance function $d$ so that $d\in\dscript (\nabla _\omega )$ then $\nabla _\omega d(a,b)=0$ which is why we call $\nabla _\omega$ the \textit{covariant} difference operator. We denote the length of $\omega$ from $a$ to $b$ by
$\lscript _a^b(\omega )$. Let $m>0$ be the mass (or energy) of a particle. Defining $\delta _\omega (a,b)=d(a,b)-\lscript _a^b(\omega )$, the \textit{free wave equation} at $(\omega ,a)$ is
\begin{equation}         
\label{eq32}
i\nabla _\omega v(ab)=m\delta _\omega (a,b)v(a,b)
\end{equation}
We call \eqref{eq32} the \textit{free} wave equation because it describes the wave amplitude of a particle that has no forces acting on it except gravity and gravity is not really a force anyway, it is geometry. Presumably, forces can be imposed by adding additional terms to \eqref{eq32}.

The motivation for \eqref{eq32} is the following. The lefthand side of \eqref{eq32} is a rate of change for $v$ and the coefficient of the righthand side is nonpositive so it corresponds to a resistance term. If $\omega$ is a geodesic, then the righthand side is zero and the particle moves unencumbered along $\omega$. Otherwise, the particle resists motion along $\omega$. The resistance is proportional to $m$ which is a measure of inertia and to the factor $\delta _\omega (a,b)$ which measures how far $\omega$ is from being a geodesic.

If a function $v$ is a solution to \eqref{eq32} for every $b\in\omega$ with $b>a$ and satisfies the \textit{initial condition} $v(a,b)=1$ for
$a\prec b$, then $v$ is a (\textit{free}) \textit{wave function at} $(\omega ,a)$. The initial condition and \eqref{eq32} determine $v$ uniquely using the next result.

\begin{lem}       
\label{lem31}
If $v$ is a wave function at $(\omega ,a)$ and $a< c\prec b$, then
\begin{equation}        
\begin{aligned} 
\label{eq33}
v(a,b)&=\frac{i\,d(a,b)}{i\,d(a,c)-m\delta _\omega (a,b)}\,v(a,c)\\\noalign{\medskip}
&=\frac{d(a,b)\sqbrac{d(a,c)-im\delta _\omega (a,b)}}{d(a,c)^2+m^2\delta _\omega (a,b)^2}\,v(a,c)
\end{aligned}
\end{equation}
\end{lem}
\begin{proof}
Applying \eqref{eq32} we have that
\begin{equation*}
i\sqbrac{d(a,c)v(a,b)-d(a,b)v(a,c)}=m\delta _\omega (a,b)v(a,b)
\end{equation*}
Solving for $v(a,b)$ we obtain \eqref{eq33}.
\end{proof}

Iterating \eqref{eq33}, we obtain $v(a,b)$ uniquely. If $\omega$ is a geodesic, we obtain a much simpler expression.

\begin{cor}       
\label{cor32}
If $\omega$ is a geodesic from $a$ to $b$, then the wave function at $(\omega ,a)$ is given by $v(a,b)=d(a,b)/d(a,c)$ where $a\prec c$,
$c\in\omega$.
\end{cor}
\begin{proof}
Applying Lemma~\ref{lem31} and the fact that a subpath of a geodesic is a geodesic, gives the result.
\end{proof}

We now discuss the quantum formalism on a $c$-causet $x$. Let $\brac{a_1,a_2,\ldots ,a_n}$ be the shell $S_k(x)$ and let
$\Omega _k(x)$ be the set of paths from the vertex labeled 1 to vertices in $S_k(x)$. For $\omega\in\Omega _k(x)$ with $a_j\in\omega$ denote the wave function at $(\omega ,1)$ by $v_k(\omega ,a_j)$. There are two levels of quantum mechanics in this formalism. The lower level which we call the \textit{hidden level} is probably inaccessible to us observationally. This is the level of particle trajectories or paths in
$\Omega _k(x)$. The \textit{path Hilbert space} is the set of complex-valued functions $\kscript _k(x)=L_2(\Omega _k(x))$ on
$\Omega _k(x)$ with the standard inner product
\begin{equation*}
\elbows{f,g}=\sum\brac{\fbar (\omega )g(\omega )\colon\omega\in\Omega _k(x)}
\end{equation*}
A wave function $v_k(\omega )=v_k(\omega ,a_j)$ depending on $\omega$ gives a vector in $\kscript _k(x)$. Letting 1 be the identically one function in $\kscript _k(x)$ define
\begin{equation*}
N=\elbows{1,v_k(\omega )}=\sum\brac{v_k(\omega )\colon\omega\in\Omega _k(x)}
\end{equation*}
We define the \textit{probability vector} $\vhat _k$ by $\vhat _k(\omega )=v_k(\omega )/N$. It follows that $\elbows{1,\vhat _k}=1$

For $A,B\subseteq\Omega _k(x)$ define the \textit{decoherence functional} \cite{bdghs03,hen09,sor94,sur11}
\begin{equation*}
D_k(A,B)=\sum\brac{\overline{\vhat _k(\omega )}\vhat _k(\omega ')\colon\omega\in A,\omega '\in B}
\end{equation*}
Then $B\mapsto D_k(A,B)$ is a complex-valued measure on $\Omega _k(x)$ satisfying\newline
$D_k\paren{\Omega _k(x),\Omega _k(x)}=1$. Moreover, $\overline{D_k(A,B)}=D_k(B,A)$ and $M_{ij}=D_k(A_i,A_j)$ is a positive definite matrix for every $A_1,\ldots ,A_s\subseteq\Omega _k(x)$ \cite{bdghs03,sor94,sur11}. Notice that
\begin{equation*}
D_k(A,B)=\elbows{\chi _A\vhat _k,\vhat _k}\elbows{\vhat _k,\chi _B\vhat _k}
\end{equation*}
where $\chi _A$ is the characteristic function of $A$. Moreover, 
\begin{equation*}
D_k\paren{\brac{\omega},\brac{\omega '}}=\overline{\vhat _k(\omega )}\vhat _k(\omega ')
\end{equation*}
The function $\mu _k\colon 2^{\Omega _k(x)}\to\real ^+$ given by
\begin{equation*}
\mu _k(A)=D_k(A,A)=\ab{\sum _{\omega\in A}\vhat _k(\omega )}^2
\end{equation*}
is the corresponding $q$-\textit{measure} on $2^{\Omega _k(x)}$. The function $\mu _k$ is not additive in general, but is grade-2 additive 
\cite{gud14,hen09,sor94}. Because $\mu _k$ is not additive we do not call $\mu _k(A)$ the quantum probability of $A$. Instead, we call
$\mu _k(A)$ the \textit{quantum propensity} of $A$. The propensity $\mu _k$ is useful for describing quantum interference. For
$A,B\subseteq\Omega _k(x)$ with $A\cap B=\emptyset$, we say that $A$ and $B$ \textit{interfere} in
$\mu _k(A\cup B)\ne\mu _k(A)+\mu _k(B)$.

The higher quantum level which we call the \textit{position level} is accessible to us by measurements and is the quantum theory usually considered. We define the \textit{position Hilbert space} by $H_k(x)=L_2(S_k(x))$ with the standard inner product. Define the vector
$v_k(a_j)\in H_k(x)$ by
\begin{align*}
v_k(a_j)&=\sum\brac{v_k(\omega ,a_j)\colon\omega\in\Omega _k(x),a_j\in\omega}
\intertext{Let}
N_1&=\sqbrac{\sum _{j=1}^n\ab{v_k(a_j)}^2}^{1/2}
\end{align*}
and normalize $v_k(a_j)$ to obtain the \textit{state} $\psi _k(a_j)=v_k(a_j)/N_1$. Thus, $\psi _k$ is a unit vector in $H_k(x)$. For
$A\subseteq S_k(x),\chi _A$ gives a projection operator $P(A)$ and the probability that the particle is in $A$ is given by
\begin{equation*}
p_k(A)=\elbows{\psi _k,P(A)\psi _k}=\sum\brac{\ab{\psi _k(a_j)}^2\colon a_j\in A}
\end{equation*}
In particular, the probability that the particle is at $a_j$ becomes
\begin{equation*}
p_k(a_j)=\elbows{\psi _k,P\paren{\brac{a_j}}\psi _k}=\ab{\psi _(a_j)}^2
\end{equation*}
For any $j<k$ we have the state $\psi _j$ defined in a similar way on $L_2\paren{S_j(x)}$. Note that the state $\psi _k$ is completely determined by the mass $m$ of the particle and the geometry (gravity) of the $c$-causet.

\section{Quantum Mechanical Examples} 
In these examples we shall only consider the position level of the quantum theory and the highest shell in the $c$-causet.

\begin{exam}{9}  
We begin with the simplest nontrivial universe; namely, a $c$-causet with shell sequence $(1,2,2)$. The labeled vertices are $(1;2,3;4,5)$. We interpret the vertex 1 as representing the entire inflationary period of a universe. The four paths in this $c$-causet are
$\omega _1\colon 1-2-4$, $\omega _3\colon 1-3-4$, $\omega _3\colon 1-2-5$ and $\omega _4\colon 1-3-5$. Notice that $\omega _1$,
$\omega _2$ and $\omega _4$ are geodesics and $K(4)=1$, $K(5)=0$. Applying Lemma~\ref{lem31}, the wave function becomes
\begin{align*}
v(\omega _1,4)&=\frac{d(1,4)}{d(1,2)}=\sqrt{5}\\\noalign{\smallskip}
v(\omega _2,4)&=\frac{d(1,4)}{d(1,3)}=\frac{\sqrt{5}}{2}\\\noalign{\smallskip}
v(\omega _4,5)&=\frac{d(1,5)}{d(1,3)}=\frac{\sqrt{8}}{2}=\sqrt{2}\\\noalign{\smallskip}
v(\omega _3,5)&=\frac{i\,d(1,5)}{i\,d(1,2)-m\delta _{\omega _3}(1,5)}=\frac{i\sqrt{8}}{i+m\paren{\sqrt{10}-\sqrt{8}}}\\\noalign{\smallskip}
&=\frac{\sqrt{8}\sqbrac{1+im\paren{\sqrt{10}-\sqrt{8}}}}{1+m^2\paren{\sqrt{10}-\sqrt{8}}^2}\\\noalign{\smallskip}
&=\frac{\sqrt{8}}{1+\paren{18-8\sqrt{5}}m^2}+ i\,\frac{\paren{4\sqrt{5}-8}m}{1+\paren{18-8\sqrt{5}}m^2}
\end{align*}
We now compute the functions $v(j)=\sum _iv(\omega _i,j)$. We have that
\begin{align*}
v(4)&=v(\omega _1,4)+v(\omega _2,4)=\tfrac{3}{2}\,\sqrt{5}\\
v((5)&=v(\omega _3,5)+v(\omega _4,5)\\
&=\sqbrac{\frac{\sqrt{18}}{1+\paren{18-8\sqrt{5}}m^2}+\sqrt{2}}+\frac{i\paren{4\sqrt{5}-8}m}{{1+\paren{18-8\sqrt{5}}m^2}+\sqrt{2}}
\end{align*}
The normalization constant becomes
\begin{equation*}
N_1^2=\ab{v(4)}^2+\ab{v(5)}^2=\frac{53}{4}+\frac{16}{1+\paren{18-8\sqrt{5}}m^2}
\end{equation*}
The state is $\psi (4)=v(4)/N_1$, $\psi (5)=v(5)/N_1$. This gives the probability
\begin{equation*}
p(4)=\ab{\psi (4)}^2
=\frac{\ab{v(4)}^2}{N_1^2}=\frac{45\sqbrac{1+\paren{18-8\sqrt{5}}m^2}}{64+53\sqbrac{1+\paren{18-8\sqrt{5}}m^2}}
\end{equation*}
Of course, $p(5)=1-p(4)$. These probabilities are actually functions of $m$ and we write $p(j,m)$. To be precise, $p(j,m)$ is not defined for $m=0$ and we write
\begin{equation*}
p(j,0)=\lim _{m\to 0}p(j,m)
\end{equation*}
If one graphs $p(4,m)$ as a function of $m$, one obtains an increasing function with values $p(4,0)=0.38462$, $p(4,2.5)=0.49601$ and
$p(4,10)=0.84901$. In fact
\begin{equation*}
\lim _{m\to\infty}p(4,m)=0.94906
\end{equation*}
This gives an unusual behavior at low mass (energy). For $m<2.5$ the particle prefers to move toward the lower curvature vertex~5 than toward the higher curvature vertex~4. This is either a strange behavior or may indicate that about $m=2.5$ is a lower bound for possible values of $m$.\qquad\qedsymbol
\end{exam}  

\begin{exam}{10}  
We now consider a slightly more complicated universe with a $c$-causet having shell sequence $(1,3,2)$. The labeled vertices are
$(1;2,3,4;5,6)$. The six paths are $\omega _1\colon 1-2-5$, $\omega _2\colon 1-3-5$, $\omega _3\colon 1-4-5$, $\omega _4\colon 1-2-6$, $\omega _5\colon 1-3-6$, $\omega _6\colon 1-4-6$. We see that $\omega _2$, $\omega _5$ and $\omega _6$ are geodesics and $K(5)=0$, $K(6)=1$. Applying Lemma~\ref{lem31}, the wave function becomes
\begin{align*}
v(\omega _1,5)&=\frac{i\,d(1,5)}{i\,d(1,2)-m\delta _{\omega _1}(1,5)}=\frac{i\sqrt{8}}{i+m\paren{\sqrt{10}-\sqrt{8}}}\\\noalign{\smallskip}
&=\frac{\sqrt{8}\sqbrac{1+im\paren{\sqrt{10}-\sqrt{8}}}}{1+m^2\paren{\sqrt{10}-\sqrt{8}}^2}\\\noalign{\smallskip}
v(\omega _2,5)&=\frac{d(1,5)}{d(1,3)}=\frac{\sqrt{8}}{2}=\sqrt{2}\\
v(\omega _3,5)&=\frac{i\,d(1,5)}{i\,d(1,4)-m\delta _{\omega _3}(1,5)}=\frac{i\sqrt{8}}{3i+m\paren{\sqrt{10}-\sqrt{8}}}\\\noalign{\smallskip}
&=\frac{\sqrt{8}\sqbrac{3+im\paren{\sqrt{10}-\sqrt{8}}}}{9+m^2\paren{\sqrt{10}-\sqrt{8}}^2}\\\noalign{\smallskip}
v(\omega _4,6)&=\frac{i\,d(1,6)}{i\,d(1,2)-m\delta _{\omega _4}(1,6)}=\frac{i\sqrt{13}}{i+m\paren{\sqrt{17}-\sqrt{13}}}\\\noalign{\smallskip}
&=\frac{\sqrt{13}\sqbrac{1+im\paren{\sqrt{17}-\sqrt{13}}}}{1+m^2\paren{\sqrt{17}-\sqrt{13}}^2}\\\noalign{\smallskip}
v(\omega _5,6)&=\frac{d(1,6)}{d(1,3)}=\frac{i\sqrt{13}}{2}\\\noalign{\smallskip}
v(\omega _6,6)&=\frac{d(1,6)}{d(1,4)}=\frac{i\sqrt{13}}{3}
\end{align*}
Hence,
\begin{align*}
v(5)&=\sqbrac{\frac{\sqrt{8}}{1+m^2\paren{\sqrt{10}-\sqrt{8}}^2}+\sqrt{2}+\frac{3\sqrt{8}}{9+m^2\paren{\sqrt{10}-\sqrt{8}}^2}}\\\noalign{\smallskip}
  &\quad+i\sqrt{8}m\paren{\sqrt{10}-\sqrt{8}}
  \sqbrac{\frac{1}{1+m^2\paren{\sqrt{10}-\sqrt{8}}^2}+\frac{1}{1+m^2\paren{\sqrt{10}-\sqrt{8}}^2}}\\\noalign{\smallskip}
  v(6)&=\sqrt{13}\sqbrac{\frac{5}{6}+\frac{1}{1+m^2\paren{\sqrt{17}-\sqrt{13}}^2}}
  +\frac{im\paren{\sqrt{17}-\sqrt{13}}}{1+m^2\paren{\sqrt{17}-\sqrt{13}}^2}
\end{align*}
As in Example~9, we can compute $\psi (5)=v(5)/N_1$ and $\psi (6)=v(6)/N_1$. We then find the probability
\begin{equation*}
p(6)=\frac{\ab{v(6)}^2}{\ab{v(5)}^5+\ab{v(6)}^2}=\frac{1}{\ab{\frac{v(5)}{v(6)}}^2+1}
\end{equation*}
As a function of $m$ we have that $p(6,0)=0.61905$ and
\begin{equation*}
\lim _{m\to\infty}p(6,m)=0.81864
\end{equation*}
The function $p(6,m)$ is essentially increasing as a function of $m$.\qquad\qedsymbol
\end{exam}  

\begin{exam}{11}  
Our final example is the $c$-causet with shell sequence $(1,2,3)$ and labeled vertices $(1;2,3;4,5,6)$. The six paths are
$\omega _1\colon 1-2-4$, $\omega _2\colon 1-3-4$, $\omega _3\colon 1-2-5$, $\omega _4\colon 1-3-5$, $\omega _5\colon 1-2-6$ and
$\omega _6\colon 1-3-6$. We see that $\omega _1$, $\omega _2$, $\omega _4$ and $\omega _6$ are geodesics and $K(4)=1$, $K(5)=K(6)=0$. Applying Lemma~\ref{lem31}, the wave function becomes
\begin{align*}
v(\omega _1,4)&=\frac{d(1,4)}{d(1,2)}=\sqrt{5}\\\noalign{\smallskip}
v(\omega _2,4)&=\frac{d(1,4)}{d(1,3)}=\frac{\sqrt{5}}{2}\\\noalign{\smallskip}
v(\omega _3,5)&=\frac{i\,d(1,5)}{i\,d(1,2)-m\delta _{\omega _3}(1,5)}=\frac{i\sqrt{8}}{i+m\paren{\sqrt{10}-\sqrt{8}}}\\\noalign{\smallskip}
&=\frac{\sqrt{8}\sqbrac{1+im\paren{\sqrt{10}-\sqrt{8}}}}{1+m^2\paren{\sqrt{10}-\sqrt{8}}^2}\\\noalign{\smallskip}
v(\omega _4,5)&=\frac{d(1,5)}{d(1,3)}=\frac{\sqrt{8}}{2}=\sqrt{2}\\\noalign{\smallskip}
v(\omega _5,6)&=\frac{i\,d(1,6)}{i\,d(1,2)-m\delta _{\omega _5}(1,6)}=\frac{i\sqrt{13}}{i+m\paren{\sqrt{17}-\sqrt{13}}}\\\noalign{\smallskip}
&=\frac{\sqrt{13}\sqbrac{1+im\paren{\sqrt{17}-\sqrt{13}}}}{1+m^2\paren{\sqrt{17}-\sqrt{13}}^2}\\\noalign{\smallskip}
v(\omega _6,6)&=\frac{d(1,6)}{d(1,3)}=\frac{\sqrt{13}}{2}
\end{align*}
We conclude that
\begin{align*}
v(4)&=\tfrac{3}{2}\sqrt{5}\\\noalign{\smallskip}
v(5)&=\sqbrac{\frac{\sqrt{8}}{1+m^2\paren{\sqrt{10}-\sqrt{8}}^2}+\sqrt{2}}
  +\frac{i\sqrt{8}\,m\paren{\sqrt{10}-\sqrt{8}}}{1+m^2\paren{\sqrt{10}-\sqrt{8}}^2}\\\noalign{\smallskip}
v(6)&=\sqrt{13}\sqbrac{\frac{1}{1+m^2\paren{\sqrt{17}-\sqrt{13}}^2}+\frac{1}{2}}
  +\frac{i\sqrt{13}\,m\paren{\sqrt{17}-\sqrt{13}}}{1+m^2\paren{\sqrt{17}-\sqrt{13}}^2}
\end{align*}
It follows that
\begin{align*}
N_1^2&=\ab{v(4)}^2+\ab{v(5)}^2+\ab{v(6)}^2\\\noalign{\smallskip}
&=\frac{33}{2}+\frac{16}{1+m^2\paren{\sqrt{10}-\sqrt{8}}^2}+\frac{26}{1+m^2\paren{\sqrt{17}-\sqrt{13}}^2}
\end{align*}
Hence,
\begin{align*}
p(4,m)&=\frac{45}{4N_1^2}\\\noalign{\smallskip}
p(5,m)&=\frac{\frac{16}{1+m^2\paren{\sqrt{10}-\sqrt{8}}^2}+2}{N_1^2}\\\noalign{\smallskip}
p(6,m)&=\frac{\frac{13}{14}+\frac{26}{1+m^2\paren{\sqrt{17}-\sqrt{13}}^2}}{N_1^2}
\end{align*}
We conclude that $p(4,0)=0.192308$, $p(5,0)=0.307692$ and $p(6,0)=0.500$. We also have that
\begin{align*}
\lim _{m\to\infty}p(4,m)=\tfrac{45}{66}=0.681818\\
\lim _{m\to\infty}p(5,m)=\tfrac{4}{33}=0.121212\\
\lim _{m\to\infty}p(6,m)=\tfrac{13}{66}=0.19697
\end{align*}
The graph of $p(4,m)$ is increasing and the graphs of $p(5,m)$ and $p(6,m)$ are both decreasing. All these examples show that, at least for masses above a certain moderate level, particles tend to move toward vertices of larger curvature.\qquad\qedsymbol
\end{exam}  

\end{document}